\documentclass[11pt,twoside]{article}
\usepackage{amssymb}
\usepackage{amsmath}
\usepackage{theorem}
\usepackage{latexsym}
\usepackage{graphicx}
\reversemarginpar
\theoremheaderfont{\bf}
\theorembodyfont{\sl}
\newtheorem{theo}{Theorem}[section]
{\theorembodyfont{\rm} }
{\theorembodyfont{\rm} \newtheorem{exa}[theo]{Example}}
\newtheorem{rem}[theo]{Remark}

{\theorembodyfont{\rm} }

{\theorembodyfont{\rm}\newtheorem{alg}[theo]{Algorithm}}
{\theorembodyfont{\rm}}
\newenvironment{proof}{{\sc Proof:}}{\mbox{}\hfill$\Box$\par}
\newcommand{\eqnref}[1]{~\mbox{$\text{(\ref{#1})}$}}

\newcommand{\junk}[1]{}

\newcommand{\TS}{\textstyle}
\newcommand{\N}{{\mathbb N}}
\newcommand{\F}{{\mathbb F}}

\newcommand{\cC}{{\mathcal C}}

\newcommand{\cB}{{\mathcal B}}
\newcommand{\cO}{{\mathcal O}}

\newcommand{\fv}{\mathfrak{v}}

\newcommand{\rank}{\text{rk}\,}

\newcommand{\im}{\mbox{\rm im}\,}

\newcommand{\dist}{\mbox{\rm dist}}
\newcommand{\dfree}{\mbox{${\rm d}_{\rm free}$}}

\newcommand{\wt}{\mbox{${\rm wt}$}}
\newcommand{\dd}{\mbox{${\rm d}$}}

\newcounter{alp}
\newcounter{ara}
\newcounter{rom}

\newenvironment{alphalist}{\begin{list}{(\alph{alp})\hfill}{\usecounter{alp}
     \topsep0ex \labelwidth.7cm \leftmargin.7cm \labelsep0cm
     \rightmargin0cm \parsep0ex \itemsep.6ex
     \partopsep1.6ex}}{\end{list}}
\newenvironment{arabiclist}{\begin{list}{(\arabic{ara})\hfill}{\usecounter{ara}
     \topsep0ex \labelwidth.7cm \leftmargin.7cm \labelsep0cm
     \rightmargin0cm \parsep0ex \itemsep.6ex
     \partopsep1.6ex}}{\end{list}}
\newenvironment{Algolist}{\begin{list}{$\bullet$\hfill}{\usecounter{ara}
     \topsep0ex \labelwidth.3cm \leftmargin.3cm \labelsep0cm
     \rightmargin0cm \parsep0ex \itemsep.6ex
     \partopsep1.6ex}}{\end{list}}
\newenvironment{algolist}{\begin{list}{$\diamond$\hfill}{\usecounter{ara}
     \topsep0ex \labelwidth.3cm \leftmargin.3cm \labelsep0cm
     \rightmargin0cm \parsep0ex \itemsep.6ex
     \partopsep1.6ex}}{\end{list}}

\title{Algebraic Decoding for Doubly Cyclic Convolutional Codes}
\date{August 4, 2009}
\author{Heide Gluesing-Luerssen\footnote{University of Kentucky, Department of Mathematics,
       715 Patterson Office Tower, Lexington, KY 40506-0027, USA; heidegl@ms.uky.edu} ,
       Uwe Helmke\thanks{Institut f\"ur Mathematik, Lehrstuhl f\"ur Mathematik~II, Universit\"at W\"urzburg,
       Am Hubland, 97074~W\"urzburg, Germany; helmke@mathematik.uni-wuerzburg.de} ,
       Jos{\'e} Ignacio Iglesias Curto\thanks{Universidad de Salamanca, Departamento de Matem{\'a}ticas,
       Plaza de la Merced 1--4, 37008 Salamanca, Spain; joseig@usal.es}
       }
\begin{document}
\maketitle

\renewcommand{\thefootnote}{}
\footnotetext{Research by U.H. has been partially
supported by grant HE 1858/12-1 within the DFG SPP 1305.
Research by J.I.I.C. has been partially supported by Junta de Castilla y Le{\'o}n through research project SA029A08}

{\bf Abstract:}
An iterative decoding algorithm for convolutional codes is presented.
It successively processes~$N$ consecutive blocks of the received word in order to decode the first block.
A bound is presented showing which error configurations can be corrected.
The algorithm can be efficiently used on a particular class of convolutional codes, known as doubly cyclic convolutional codes.
Due to their highly algebraic structure those codes are well suited for the algorithm and the main step of the procedure
can be carried out using Reed-Solomon decoding.
Examples illustrate the decoding and a comparison with existing algorithms is being made.



\section{Introduction}
\setcounter{equation}{0}

The main task of coding theory can be described as designing codes with good error-correcting performance along with an
efficient decoding algorithm.
For block codes two types of answers are known to this quest.
On the one hand, there are algebraic decoding algorithms for the special class of BCH codes, see, e.g.,~\cite[Sec.~5.4]{HP03},
including the more recent list decoding procedures as developed in~\cite{Su97,GuSu99}, see also~\cite{RoRu00}.
On the other hand, there are graph-based decoding algorithms as introduced by~\cite{Wo78}.
These are iterative methods and work particularly well for LDPC codes; for an overview see for instance the thesis~\cite{Wi96}.

For convolutional codes the most prominent decoding algorithms are the Viterbi algorithm~\cite{Vi67} and variants thereof.
They are all trellis-based algorithms and mainly applicable to codes over small alphabets and not too large degree in
order to keep the underlying graph at a manageable size; for an overview see the monographs~\cite{JoZi99,LC83}.
In the 1970's the first attempts were made in order to construct convolutional codes with some additional underlying algebraic
structure in the hope of decoding them algebraically~\cite{Ju75,Pi76,Ro79}.
Later constructions included BCH convolutional codes~\cite{RoYo99} as well as specific constructions of cyclic convolutional
codes~\cite{GS04,GS06p} and Goppa convolutional codes~\cite{MPCS06}.
In the paper~\cite{Ro99}, a first step toward an algebraic decoding algorithm for convolutional codes has been made.
It is based on an input/state/output description of the code and relies on the controllability matrix being the parity check
matrix of an algebraically decodable block code.
This makes the algorithm particularly suitable for the BCH codes developed in~\cite{RoYo99}.
In the thesis~\cite[Sec.~4.2]{Ts08} it is shown that the algorithm also applies to a certain class of $1$-dimensional
cyclic convolutional codes appearing as a special case in~\cite{GS06p}.
Finally, in the thesis~\cite[Sec.~4.3]{Wei98} a decoding algorithm for unit memory convolutional codes is developed which may be turned into
an algebraic algorithm if the underlying block codes can be decoded algebraically.
We will return to these algorithms at the end of Section~\ref{S-ExaComp} when comparing the performance of our
algorithm with theirs.

In this paper we will present an algebraic decoding algorithm for a particular class of convolutional codes.
It will depend on a chosen parameter~$N$ and a certain weight bound~$d$ and can correct up to $\lfloor d/2\rfloor$
errors appearing on any time window of length~$N$.
The algorithm is a special version of a decoding algorithm appearing first in~\cite[Sec.~4.4]{IC08} and~\cite{HI09}.
As opposed to our presentation, the algorithm is cast completely in the setting of input/state/output descriptions
in~\cite{IC08,HI09}.
The procedure works sequentially in the sense that~$N$ consecutive blocks of the received word are processed in order to
decode the first of those blocks.
This decoding step is based on the partial decoding of a certain block code.
Thereafter the algorithm moves one block further.
The details, in a slightly more general version, will be presented in the next section.
In Section~\ref{S-BlockAlg} we will show that the class of doubly cyclic convolutional codes, introduced in~\cite{GS06p}, is
particularly well suited for this algorithm.
Indeed, first of all the error parameter~$d$ can be made quite large (compared to the length, dimension, and degree of the code), and
secondly the partial decoding of the underlying block code can be achieved by the well-known and efficient Reed-Solomon decoding.
It should be mentioned that due to the large field size and degree of doubly cyclic codes, Viterbi decoding is not feasible
for this particular class of convolutional codes.
In the final section we will run some detailed examples and will compare our algorithm to the decoding algorithms
mentioned above with respect to error-correcting performance and time complexity.

Let us close the introduction with recalling the basic notions of convolutional coding theory as needed throughout the paper.
Let~$\F$ be a finite field and let $\F[z]$ and $\F[\![z]\!]$ denote the
rings of polynomials and formal power series in~$z$, respectively.
Throughout this paper we regard vectors as row vectors, so that in every vector-matrix multiplication the
matrix appears on the right.
A {\sl convolutional code of length\/} $n$ is an $\F[\![z]\!]$-submodule~$\cC$ of $\F[\![z]\!]^n$ of the form
\[
    \cC=\im G:=\{uG\,\big|\, u\in\F[\![z]\!]^k\}
\]
where~$G$ is a {\sl basic\/} matrix in $\F[z]^{k\times n}$, i.~e.
$\rank G(\lambda)=k \text{ for all }\lambda\in\overline{\F}$,
with $\overline{\F}$ being an algebraic closure of~$\F$.
We call such a matrix $G$ an {\sl encoder}, and the number
$\deg(\cC):=\deg(G):=\max\{\deg(M)\mid M\text{ is a $k$-minor of }G\}$
is said to be the {\sl degree\/} of the encoder~$G$ or of the code~$\cC$.
For each basic matrix the sum of its row degrees is at least $\deg(G)$, where the degree of a
polynomial row vector is defined as the maximal degree of its entries.
A matrix $G\in\F[z]^{k\times n}$ is said to be {\sl reduced\/} if the sum of its row degrees equals
$\deg(G)$; for the many characterizations of reducedness see, e.~g., \cite[Main~Thm.]{Fo75} or \cite[Thm.~A.2]{McE98}.
It is well known~\cite[p.~495]{Fo75} that each convolutional code admits a reduced encoder.
The row degrees of a reduced encoder are, up to ordering, uniquely determined by the code and are called the
{\sl Forney indices\/} of the code or of the encoder, and the maximal Forney index is called the {\sl memory\/} of the code.
The main example class of convolutional codes in this paper, so called doubly-cyclic convolutional codes, will be introduced in
Theorem~\ref{P-DCCC}.

Besides these algebraic notions the main concept in error-control coding is the weight.
The well-known {\sl Hamming weight\/} of a vector $v=(v_1,\ldots,v_n)\in~\F^n$ is given as
$\wt(v)=\#\{i\mid v_i\not=0\}$ and $\dd(v,\,w):=\wt(v-w)$ denotes the associated Hamming distance.
For a polynomial vector $v=\sum_{t=0}^N v_t z^t\in\F[z]^n$, $v_t\in\F^n$, we define
its {\sl overall Hamming weight\/} as $\wt(v)=\sum_{t=0}^N \wt(v_t)$.
The {\sl distance\/} $\min\{\wt(v)\,|\, v\in\cC,\,v\not=0\}$ of a block code~$\cC$ is
denoted by $\dist(\cC)$, while for a convolutional code it is written as $\dfree(\cC)$.

\section{A General Decoding Algorithm}\label{S-GenAlg}
\setcounter{equation}{0}
We consider a convolutional code $\cC=\im G:=\{uG\mid u\in\F[\![z]\!]^k\}\subseteq\F[\![z]\!]^n$ having a
basic generator matrix
\begin{equation}\label{e-G}
  G=\sum_{j=0}^m G_{j}z^{j}\in\F[z]^{k\times n},\ G_{j}\in\F^{k\times n}.
\end{equation}
For the decoding algorithm we need to fix a processing depth $N\in\N$ and define the block code
\begin{equation}\label{e-Ghat}
  \cB:=\im \hat{G}\subseteq\F^{Nn},\text{ where }
  \hat{G}:=\begin{pmatrix}G_0&G_1&\ldots&G_{N-1}\\ &G_0&\ldots&G_{N-2}\\
     & &\ddots&\vdots\\ & & &G_0\end{pmatrix}\in\F^{Nk\times Nn},
\end{equation}
where, as usual, $G_{j}=0$ for $j>m$ and the empty triangular part is filled with zero entries.
Notice that due to the basicness of the encoder~$G$, the matrix~$G_0$, and hence~$\hat{G}$, has full row rank.

Besides the processing depth~$N$, the decoding algorithm will also depend on the choice of a step size
parameter~$L\in\{1,\ldots,N\}$ and a weight parameter $d:=d(L)\geq\dist(\cB)-1$ satisfying
\begin{equation}\label{e-errorass}
  v:=(v_0,\ldots,v_{N-1})\in\cB,\
  \wt(v)\leq d\Longrightarrow (v_0,\ldots,v_{L-1})=0.
\end{equation}
It is clear that, $d=\dist(\cB)-1=\dist(\im G_0)-1$ satisfies\eqnref{e-errorass}, regardless of the value of~$L$.
Later on we will see that the error-correcting bound of our decoding algorithm will be given by
$\lfloor d/2\rfloor$ and therefore we will be interested in choosing~$d$ as large as possible.
However, Algorithm~\ref{A-Dec} below will not depend on choosing~$d$ optimal.
In the next section, we will present a class of codes along with a specific large weight parameters~$d$
satisfying\eqnref{e-errorass}, and we will show how to carry out the main step of the algorithm for those
codes efficiently.

\begin{exa}\label{e_dchoice}
Let $G=(1,\,1,\,z,\,\ldots,\,z)\in\F_2[z]^{1\times n}$. Choose $N=2$ and $L=1$. Then
\[
  \cB=\im\left(\!\!\begin{array}{ccccc|ccccc}1&1&0&\ldots&0&0&0&1&\ldots&1\\ \hline 0&0&0&\ldots&0&1&1&0&\ldots&0\end{array}\!\!\right).
\]
Notice that~$\dist(\cB)=\dist(\im G_0)=2$.
By inspecting all~$4$ codewords in~$\cB$ we see that $d=n-1$ is the largest value for which\eqnref{e-errorass}
is true
(actually, $\wt(v)\leq n-1$ implies $\wt(v)\leq2$).
One should also observe that in this case $d=\dfree(\cC)-1$, which is the maximum possible value for~$d$, see Remark~\ref{R-vweight}(b).
\end{exa}

We will also need the matrix $\tilde{G}\in\F^{mk\times Nn}$ defined as
\[
  \tilde{G}:=\!\!\begin{pmatrix}G_m& & &\\G_{m-1}&G_m& & \\ \vdots&\vdots&\ddots& \\
     G_{m-N+1}&\!\!G_{m-N+2}\!\!&\ldots& G_m \\
     \vdots&\vdots& & \vdots\\
    G_1&G_2&\ldots&G_N\end{pmatrix}\!\!\text{ if $N\leq m$ and }
  \tilde{G}:=\!\!\begin{pmatrix}G_m& & & &0\\G_{m-1}\!\!&G_m& & &0 \\ \vdots&\vdots&\ddots& &\vdots \\
    G_1&G_2&\ldots&G_m&0\end{pmatrix}\!\!\text{ if $N>m$,}
\]
where in the second case the zero matrices on the very right consist of $(N-m)n$ columns.
Notice that the matrix $\Big(\begin{smallmatrix}\tilde{G}\\ \hat{G}\end{smallmatrix}\Big)\in\F^{(N+m)k\times Nn}$
is a typical block in the sliding generator matrix of the code~$\cC$.
In particular, if $v=\sum_{t\geq0}v_tz^t=(\sum_{t\geq0}u_tz^t)G$ is a codeword in~$\cC$, then
\begin{equation}\label{e-vu}
\left.
\begin{array}{rcl}
  (v_j,v_{j+1},\ldots,v_{j+N-1})&\!\!\!\!=\!\!\!\!&(u_j,u_{j+1},\ldots,u_{j+N-1})\hat{G}+(u_{j-m},u_{j-m+1},\ldots,u_{j-1})\tilde{G}\\[1ex]
  &\!\!\!\!=\!\!\!\!&(u_{j-m},u_{j-m+1},\ldots,u_{j+N-1})\begin{pmatrix}\tilde{G}\\\hat{G}\end{pmatrix}
\end{array}
\hspace*{2em}\right\}
\end{equation}
for all $j\geq0$.
For a power series $v=\sum_{t\geq0}v_tz^t$ and $M\in\N$
we denote by $v_{[0,M]}$ the truncation $\sum_{t=0}^M v_tz^t$ of~$v$ at time~$M$.

\begin{rem}\label{R-vweight}
Let~$L$ and~$d$ satisfy\eqnref{e-errorass}.
\begin{alphalist}
\item From the previous discussion it follows that if $v\in\cC$ is such that $\wt(v_{[jL,jL+N-1]})\leq d$ for all $j\geq0$,
      then $v=0$.
      Indeed, suppose $v=\sum_{t\geq0}v_tz^t=(\sum_{t\geq0}u_tz^t)G$. Then
      $(v_0,\ldots,v_{N-1})=(u_0,\ldots,u_{N-1})\hat{G}\in\cB$ and thus\eqnref{e-errorass} implies $v_j=0$ for $j=0,\ldots,L-1$.
      Hence $u_j=0$ for $j=0,\ldots,L-1$ due to the full row rank of~$G_0$ (delay-freeness of~$G$) and Equation\eqnref{e-vu} shows that
      $(v_L,\ldots,v_{L+N-1})=(u_L,u_{L+1},\ldots,u_{L+N-1})\hat{G}\in\cB$ and\eqnref{e-errorass} implies $v_j=0$ for $j=L,\ldots,2L-1$.
      Proceeding this way leads to $v=0$.
      This also shows that if~$d$ is the largest value satisfying\eqnref{e-errorass} for $L=1$, then
      $d+1$ is the $(N-1)$-th column distance of the convolutional code~$\cC$ in the sense of
      \cite[Sec.~3.1]{JoZi99}.
\item The free distance of the code is at least~$d+1$. Indeed, suppose $v\in\cC\backslash\{0\}$ and $j\in\N_0$ is minimal such
      that $v_j\not=0$. Then $(v_j,\ldots,v_{j+N-1})\in\cB$ and $\wt(v)\geq\wt(v_j,\ldots,v_{j+N-1})\geq d+1$ by\eqnref{e-errorass}.
\end{alphalist}
\end{rem}

Now we are ready to formulate the general steps of the decoding algorithm.
The algorithm has been presented first in~\cite[Sec.~4.4]{IC08} and~\cite{HI09}, where it is given in a more
general form and within the context of the tracking problem of control theory.
In those papers it is given in terms of an input/state/output representation of the convolutional code.

Let us fix $N\in\N$ and $L,\,d$ satisfying\eqnref{e-errorass}.
The following algorithm will, in each cycle, process strings of~$N$ consecutive received blocks in order to decode the
first~$L$ of those blocks with respect to the convolutional code~$\cC=\im G$.
In the next cycle the algorithm will move~$L$ steps further down the time axis.

\begin{alg}\label{A-Dec}
Let $\sum_{t\geq0}\tilde{v}_tz^t\in\F[\![z]\!]^n$ be a received word.
\\
Suppose that for some $j\geq0$ we have computed $\hat{u}_t\in\F^k,\,\hat{v}_t\in\F^n,\,t=0,\ldots,jL-1$.
We assume that~$\hat{v}$ is the decoding of $\tilde{v}$ on the time interval
$[0,\,jL-1]$ and that $\hat{u}$ is the associated message string.
In the initial step where $j=0$ this condition is empty and in Step~1 the vector~$S$ is set to zero.
\\[.5ex]
\underline{Step 1:}
Put $\tilde{V}:=(\tilde{v}_{jL},\ldots,\tilde{v}_{jL+N-1})$ and
$\hat{S}:=(\hat{u}_{jL-m},\hat{u}_{jL-m+1},\ldots,\hat{u}_{jL-1})\tilde{G}$ (where $\hat{u}_i=0$ for $i<0$).
Decode the word $\tilde{w}:=\tilde{V}-\hat{S}$ with respect to the code~$\cB$ in such a way that
if $\dd(\tilde{w},\cB)\leq\lfloor d/2\rfloor$ then the decoded word~$\hat{w}\in\cB$ satisfies
$\dd(\tilde{w},\hat{w})\leq\lfloor d/2\rfloor$
(if $\dd(\tilde{w},\cB)>\lfloor d/2\rfloor$, no specification is made for the decoded word $\hat{w}\in\cB$).
Let $\hat{u}:=(\hat{u}_{jL},\ldots,\hat{u}_{jL+N-1})\in\F^{Nk}$ be the message associated to~$\hat{w}$, that is,
$\hat{w}=\hat{u}\hat{G}$.
Put
\begin{equation}\label{e-uyhat}
  \hat{w}+\hat{S}=:(\hat{v}_{jL},\hat{v}_{jL+1},\ldots,\hat{v}_{jL+N-1})\in\F^{Nn}
\end{equation}
Return the data $\hat{u}_t,\,\hat{v}_t,\,t=jL,\ldots,(j+1)L-1$ as the decoding of $\tilde{v}_t$
on the time interval $[jL,(j+1)L-1]$ and discard the remaining entries of~$\hat{w}+\hat{S}$ and~$\hat{u}$.
\\[.5ex]
\underline{Step 2:}
Replace~$j$ by $j+1$ and return to Step~1.
\end{alg}

\begin{theo}\label{T-DecAlg}
Suppose the codeword $v=\sum_{t\geq0}v_tz^t=(\sum_{t\geq0}u_tz^t)G\in\cC$ has been sent and the word
$\sum_{t\geq0}\tilde{v}_tz^t\in\F[\![z]\!]^n$ has been received.
\begin{arabiclist}
\item The data returned by Algorithm~\ref{A-Dec} satisfy $\hat{v}:=\sum_{t\geq0}\hat{v}_tz^t=(\sum_{t\geq0}\hat{u}_tz^t)G$,
      thus $\hat{v}$ is a codeword in~$\cC$ with associated message~$\sum_{t\geq0}\hat{u}_tz^t$.
\item Let~$L$ and~$d$ satisfy\eqnref{e-errorass}. If the transmission errors satisfy
      \begin{equation}\label{e-error}
         \dd\big((v_{jL},v_{jL+1},\ldots,v_{jL+N-1}),(\tilde{v}_{jL},\tilde{v}_{jL+1},\ldots,\tilde{v}_{jL+N-1})\big)
         \leq \Big\lfloor\frac{d}{2}\Big\rfloor\text{ for all }j\geq 0,
      \end{equation}
      then $\hat{v}=v$, that is, Algorithm~\ref{A-Dec} will return the sent codeword~$v$.
      In particular, $\dd(\tilde{w},\cB)\leq\lfloor\frac{d}{2}\rfloor$ in each cycle of Step~1).
\end{arabiclist}
\end{theo}

Notice that, due to Remark~\ref{R-vweight}(a), for any received word $\tilde{v}\in\F[\![z]\!]^n$ there exists at most one codeword
$v\in\cC$ satisfying\eqnref{e-error}.

\begin{proof}
(1) Assume that for some $j\geq0$ we have already computed the data $\hat{u}_t,\,\hat{v}_t,\,t=0,\ldots,jL-1$
and that
\begin{equation}\label{e-stepass}
   \sum_{t=0}^{jL-1}\hat{v}_tz^t=\Big((\sum_{t=0}^{jL-1}\hat{u}_tz^t)G\Big)_{[0,jL-1]},
\end{equation}
which, for $j=0$, is an empty assumption.
The next step of the algorithm produces
\[
  (\hat{v}_{jL},\ldots,\hat{v}_{jL+N-1})=(\hat{u}_{jL},\ldots,\hat{u}_{jL+N-1})\hat{G}
  +(\hat{u}_{jL-m},\hat{u}_{jL-m+1}\ldots,\hat{u}_{jL-1})\tilde{G},
\]
see\eqnref{e-uyhat}.
Hence $\hat{v}_{jL+t}=\sum_{i=0}^{t}\hat{u}_{jL+i}G_{t-i}+\sum_{i=0}^{m-t-1}\hat{u}_{jL-1-i}G_{t+1+i}$, where the second sum is zero if $t\geq m$.
In either case, we derive $\hat{v}_{jL+t}=\sum_{i=0}^m \hat{u}_{jL+t-i}G_i$ for $t=0,\ldots,N-1$ and
together with\eqnref{e-stepass} this shows that
$\sum_{t=0}^{jL+N-1}\hat{v}_tz^t=\Big((\sum_{t=0}^{jL+N-1}\hat{u}_tz^t)G\Big)_{[0,jL+N-1]}$.
In particular,~\eqnref{e-stepass} is true for $j+1$ instead of~$j$ (recall that the algorithm only returns $v_{jL},\ldots,v_{(j+1)L-1}$).
This completes the proof of~(1).
\\
(2)
Suppose that for some fixed $j\geq0$ the algorithm correctly returned
$\hat{v}_t=v_t$ and $\hat{u}_t=u_t$ for all $t\leq jL-1$.
Put $V:=(v_{jL},\ldots,v_{jL+N-1})$.
Write $\tilde{w}=(\tilde{w}_0,\ldots,\tilde{w}_{N-1})$ and $\hat{S}=(\hat{S}_0,\ldots,\hat{S}_{N-1})$ for the data in Step~1 of the algorithm.
By\eqnref{e-vu} we have
\[
   V=(u_{jL},\ldots,u_{jL+N-1})\hat{G}+(u_{jL-m},\ldots,u_{jL-1})\tilde{G}=(u_{jL},\ldots,u_{jL+N-1})\hat{G}+\hat{S},
\]
where the second identity follows from $u_t=\hat{u}_t$ for $t\leq jL-1$.
Hence $V-\hat{S}\in\cB$.
The error assumption\eqnref{e-error} implies $\dd\big((\tilde{V}-\hat{S}),(V-\hat{S})\big)=\dd(\tilde{V},V)\leq\lfloor d/2\rfloor$.
Thus, $\dd(\tilde{w},\cB)\leq\lfloor d/2\rfloor$ for $\tilde{w}=\tilde{V}-\hat{S}$
and the decoding requirement made in Step~1)
implies $\dd(\hat{w},\tilde{w})\leq\lfloor d/2\rfloor$ for the decoded word
$\hat{w}=(\hat{w}_0,\ldots,\hat{w}_{N-1})=(\hat{u}_{jL},\ldots,\hat{u}_{jL+N-1})\hat{G}\in\cB$.
As a consequence, $\dd(\hat{w},V-\hat{S})\leq d$ and\eqnref{e-errorass} implies $\hat{w}_{t}=v_{jL+t}-\hat{S}_t$ for
$t=0,\ldots,L-1$. But then\eqnref{e-uyhat} yields that $v_{jL+t}=\hat{v}_{jL+t}$ for $t=0,\ldots,L-1$.
Finally, the uniqueness of the associated message sequence (or the full row rank of~$\hat{G}$) implies
$u_{jL+t}=\hat{u}_{jL+t}$ for $t=0,\ldots,L-1$.
\end{proof}

Notice that the algorithm will, in each cycle, decode a string of~$L$ consecutive codeword blocks.
The most interesting case will be $L=1$, which will lead to a possibly larger~$d$ satisfying\eqnref{e-errorass} and
thus to a larger amount of errors that can be corrected.
In the next section we will concentrate on that case.

\section{Partial Decoding of the Block Code ${\mathcal B}$}\label{S-BlockAlg}
\setcounter{equation}{0}
The main step of Algorithm~\ref{A-Dec} is the partial decoding with respect to the block code~$\cB$ from\eqnref{e-Ghat}.
In this section we will show how to carry out this step efficiently for a particular class of convolutional codes,
which were designed in~\cite{GS06p}.
For those codes the decoding step essentially amounts to decoding certain Reed-Solomon block codes.
The codes can be defined as follows.

Let $\F=\F_q$ be a field with~$q$ elements and primitive element~$\alpha$.
Put $A=\F[x]/_{(x^n-1)}$, where $n:=q-1$, and let
$\fv:\; A\longrightarrow \F^n,\ \sum_{i=0}^{n-1}f_i x^i\longmapsto (f_0,\ldots,f_{n-1})$
be the canonical vector space isomorphism between~$A$ and~$\F^n$.
Fix $k\in\{0,\ldots,n-1\}$ and consider the $\F$-algebra automorphism $\sigma\,: A\longrightarrow A$ defined by
$\sigma(x)=\alpha^kx$.
It is easy to see that this does indeed define an $\F$-algebra automorphism on~$A$.
The following has been shown in \cite[Exa.~3.2, Thm.~3.3, Thm.~4.3, Lem.~3.5 and its proof]{GS06p}.
\begin{theo}\label{P-DCCC} Put $n:=q-1$ and let $k\leq\lfloor n/2\rfloor$ and $m\leq\lfloor n/k\rfloor-1$.
Put $f:=\prod_{i=0}^{n-k-1}(x-\alpha^i)\in A$ and define
\[
   G_j=\begin{pmatrix}\fv\big(\sigma^j(f)\big)\\ \fv\big(\sigma^j(xf)\big)\\ \vdots\\
           \fv\big(\sigma^j(x^{k-1}f)\big)\end{pmatrix}\in\F^{k\times n}.
\]
Then
\begin{arabiclist}
\item The matrix $G=\sum_{j=0}^m G_j z^j\in\F[z]^{k\times n}$ is basic and reduced with all Forney indices equal to~$m$.
      In particular,~$m$ is the memory of the code.
\item The free distance of the convolutional code $\cC=\im G\subseteq\F[z]^n$ is $\dfree(\cC)=(m+1)(n-k+1)$.
\item For $j=0,\ldots,m$ the block code $\cB_j:=\im G_{j,0}\subseteq\F^n$, where
      \begin{equation}\label{e-Gji}
         G_{j,0}:=\begin{pmatrix}G_j\\G_{j-1}\\ \vdots\\ G_0\end{pmatrix}\in\F^{(j+1)k\times n},
      \end{equation}
      is a Reed-Solomon code of dimension $(j+1)k$ with generator polynomial
      $\prod_{i=0}^{n-(j+1)k-1}(x-\alpha^i)$.
      In particular, $\cB_j$ has distance $d_j:=n-(j+1)k+1$.
      Notice that $d_j\geq1$ for all~$j$ due to $j\leq m\leq \lfloor n/k\rfloor-1$.
\end{arabiclist}
The code $\cC=\im G\subseteq\F[\![z]\!]^n$ is called a doubly cyclic convolutional code.
\end{theo}

Let us briefly comment on the notion of cyclicity.
The convolutional code~$\cC$ is a cyclic convolutional code in the sense of \cite{GS04}.
Indeed, it can be shown that~$\cC$ may be identified with the left ideal generated by the polynomial
$g:=\sum_{j=0}^mz^j\sigma^j(f)$ in the skew-polynomial ring $A[z;\sigma]$, see also \cite[p.~165]{GS06p}.
Due to the additional cyclic structure of the block codes $\cB_j$ these codes have been named doubly cyclic in~\cite{GS06p}.
The description as left ideals in $A[z;\sigma]$, however, is not needed for this paper.
It is worth mentioning that for $k=1$, part~(2) of the theorem above shows that the codes satisfy the generalized
Singleton bound for convolutional codes~\cite{RoSm99} and thus are MDS codes.
In~\cite[p.~162]{GS06p} it has been shown that for $k=2$ the codes attain the Griesmer bound.
Thus, the codes have the best possible distance among all codes of the same length, dimension, degree, and field size;
for the Griesmer bound see~\cite[Sec.~3.5]{JoZi99} for the binary case and~\cite[Thm.~3.4]{GS03} for the general case.

We will consider the decoding algorithm of the previous section with processing depth $N:=m+1$.
Thus,
\begin{equation}\label{e-tildeB}
 \cB:=\im\hat{G}, \text{ where }
 \hat{G}:=\begin{pmatrix}G_0&G_1&\ldots&G_m\\ &G_0&\ldots&G_{m-1}\\
     & &\ddots&\vdots\\ & & &G_0\end{pmatrix}\in\F^{(m+1)k\times(m+1)n}.
\end{equation}
Due to the full row rank of the matrices in\eqnref{e-Gji} this code has the following property.
If $v=(v_0,\ldots,v_m)\in\cB$ such that $v_0=\ldots=v_{L-1}=0\not=v_{L}$ for some~$L$, then
$v_j\not=0$ for $j=L,\ldots,m$ and $\wt(v)\geq \sum_{j=0}^{m-L}d_j$, where~$d_j$ is as in Theorem~\ref{P-DCCC}(3).
As a consequence,\eqnref{e-errorass} turns into the following property.
\begin{rem}\label{R-WeightAss}
Let $L\in\{1,\ldots,m+1\}$ and $v:=(v_0,\ldots,v_m)\in\cB$.
Then $\wt(v)\leq \sum_{j=0}^{m-L+1}d_j -1$ implies $(v_0,\ldots,v_{L-1})=0$.
In particular, for $L=1$ we have
\begin{equation}\label{e-d}
  \wt(v)\leq d\Longrightarrow v_0=0,
\end{equation}
where
\begin{equation}\label{e-d2}
    d:=\sum_{j=0}^m d_j -1.
\end{equation}
It is worth mentioning that in concrete examples,~$d$ as in\eqnref{e-d2} might not be the largest value satisfying\eqnref{e-d}.
Indeed, for $\F=\F_7,\,n=6,\,k=2$, and $m=2$ one has $d_0+d_1+d_2-1=8$, but using some weight-computing routines one can show that
the largest~$d$ satisfying\eqnref{e-d} is~$10$.
However, Algorithm~\ref{A-DecBlock} presented below will be able to correct $\lfloor d/2\rfloor$ errors, where~$d$
is as in\eqnref{e-d2}.
Therefore we will not be concerned with optimizing the value of~$d$.
Notice also that
\begin{equation}\label{e-ddfree}
    d=(m+1)(n-k+1)-k\frac{(m+1)m}{2}-1=\dfree(\cC)-k\frac{m(m+1)}{2}-1.
\end{equation}
\end{rem}

Let us now turn to Algorithm~\ref{A-Dec}.
The main part in Step~1) consists of achieving the following task:
given a received vector $\tilde{v}:=(\tilde{v}_0,\ldots,\tilde{v}_m)\in\F^{(m+1)n}$ satisfying
$\dd(\tilde{v},\cB)\leq\lfloor d/2\rfloor$, return a vector $(\hat{v}_0,\ldots,\hat{v}_{L-1})$
for which there exists an extension $\hat{v}:=(\hat{v}_0,\ldots,\hat{v}_m)\in\cB$
satisfying $\dd(\hat{v},\tilde{v})\leq\lfloor d/2\rfloor$.
In the following algorithm we will carry this out for step size $L=1$.
Recall from Remark~\ref{R-WeightAss} that the parameter~$d$ from\eqnref{e-errorass} can be made largest for $L=1$
and therefore this will allow us to correct the largest amount of errors.

Throughout the rest of the paper, the phrase Reed-Solomon decoding will refer to any of the algebraic decoding algorithms for Reed-Solomon codes that
correct up to~$t$ errors, where~$t$ is the error-correcting bound of the code.
If such decoding is not possible, the algorithm returns an error message.

\begin{alg}\label{A-DecBlock}
Let the data be as in Theorem~\ref{P-DCCC} and\eqnref{e-tildeB}
and let $\tilde{v}:=(\tilde{v}_0,\ldots,\tilde{v}_m)\in\F^{(m+1)n}$.
\\[.5ex]
Put $l=m+1$.
\\[.5ex]
\hspace*{.5cm}\begin{minipage}{15.5cm}
   \underline{Step 1:} $l:= l-1$.
   \\[.5ex]
   \underline{Step 2:}
   Use Reed-Solomon decoding to decode $\tilde{v}_l$ with respect to the block code~$\cB_l$.
   If $\dd(\tilde{v}_l,\im G_{l,0})>\lfloor(d_l-1)/2\rfloor$, that is, Reed-Solomon decoding is not possible, go to Step 1.
   Else denote the resulting codeword by $w^{(l)}_l\in\cB_l=\im G_{l,0}$ and
   let $w^{(l)}_l=(\hat{x}_0,\ldots,\hat{x}_l)G_{l,0}$.
   Compute
   \begin{equation}\label{e-wl}
      w^{(l)}=(w^{(l)}_0,w^{(l)}_1,\ldots,w^{(l)}_l):=
         (\hat{x}_0,\ldots,\hat{x}_l)
         \begin{pmatrix}G_0&G_1&\ldots&G_l\\ &G_0&\ldots&G_{l-1}\\
         & &\ddots&\vdots\\ & & &G_0\end{pmatrix}.
   \end{equation}
   \underline{Step 3:} If $\dd\big(w^{(l)},(\tilde{v}_0,\ldots,\tilde{v}_l)\big)\leq\lfloor (\sum_{i=0}^l d_i -1)/2\rfloor$,
   then return $w^{(l)}_0$ and~$\hat{x}_0$.
   Else go to Step~1.
\end{minipage}
\\[.5ex]
If none of these steps yields a return, that is, $l=0$ and $\dd(w^{(0)},\tilde{v}_0)>\lfloor (d_0 -1)/2\rfloor$,
then do the following:
\\[.5ex]
\underline{Case~a):} Suppose Step~2 has been executed at least once.
Then, for each of the partial codewords $w^{(l)}$ in\eqnref{e-wl} produced in the various cycles of Step~2
use list decoding with respect to the code $\im G_0$ in order to find a codeword $y_{l+1}=\hat{x}_{l+1}G_0\in\im G_0$
that is closest to $\tilde{v}_{l+1}-\sum_{i=0}^l\hat{x}_i G_{l+1-i}$.
Set $l$ to $l+1$ and proceed the same way until $l=m$.
Put $\bar{w}^{(l)}=(\hat{x}_0,\ldots,\hat{x}_m)\hat{G}\in\cB$.
Among all the codeword $\bar{w}^{(l)}\in\cB$ produced this way choose one closest to~$\tilde{v}$, say
$w'=(w'_0,\ldots,w'_{m})=(x'_0,\ldots,x'_m)\hat{G}$, and return $w'_0$ and $x'_0$.
\\
\underline{Case~b):} Suppose Step~2 has never been executed. In this case, $\dd(\tilde{v}_l,\im G_{l,0})>\lfloor(d_l-1)/2\rfloor$
for all $l=0,\ldots,m$.
Use list decoding to produce a codeword $w^{(0)}=\hat{x}_0G_0\in\im G_0$ closest to $\tilde{v}_0$.
Set $l=0$ and proceed as in Case~a).
\end{alg}

In the next theorem we will see that Case~a) or~b) will only be invoked if no codeword~$v\in\cB$ satisfies
$\dd(\tilde{v},v)\leq\lfloor d/2\rfloor$.
When calling Algorithm~\ref{A-DecBlock} in Step~1) of Algorithm~\ref{A-Dec} this amounts to the fact that no convolutional
codeword~$v\in\cC$ satisfies the familiar error assumption\eqnref{e-error}.
Of course, if $\dd(\tilde{v},\,\cB)>\lfloor d/2\rfloor$, there are various options of how to proceed.
The easiest and cheapest solution would be to simply return any codeword $w_0=\hat{x}_0G_0$ along with its message~$\hat{x}_0$.
The strategy outlined in Algorithm~\ref{A-DecBlock} requires more effort and is designed to result in a codeword that is more likely
to be close (or even closest) to~$\tilde{v}$.
Indeed, notice that, by construction, the words
\[
  {\TS \bar{w}^{(l)}=(w^{(l)}_0,\,w^{(l)}_1,\,\ldots,\,w^{(l)}_l,\,y_{l+1}+\sum_{i=0}^l\hat{x}_i G_{l+1-i},\,\ldots,\,y_m+\sum_{i=0}^{m-1}\hat{x}_i G_{m-i})}
\]
are codewords in~$\cB$ for which the last $m-l+1$ blocks~$\bar{w}^{(l)}_i$ are codewords in~$\cB_i$ close to
the corresponding block $\tilde{v_i}$.
However, this does not guarantee that the chosen codeword will be closest to~$\tilde{v}$.
We would also like to point out that for the list decoding, see~\cite{Su97,GuSu99,RoRu00}, used in Cases~a) and~b) one might have to
increase successively the list size in order to have a nonempty return.
If more than one codeword~$y_{l+1}$ is returned one could even use all of them and extend them in the described way.
Finally it is worth mentioning that in Step~2) of the algorithm one could also replace Reed-Solomon decoding by list decoding in order to
produce a bigger pool of codewords and enhance the chances of early success in Step~3.

\begin{theo}\label{T-DecBlock}
Let~$d$ be as in\eqnref{e-d2}.
Suppose $v=(v_0,\ldots,v_m)\in\cB$ has been sent and $\tilde{v}:=(\tilde{v}_0,\ldots,\tilde{v}_m)\in\F^{(m+1)n}$ has been received.
If $\dd(v,\tilde{v})\leq \lfloor d/2\rfloor$, then there exists $l\in\{m,m-1,\ldots,0\}$ such that Step~2 of
Algorithm~\ref{A-DecBlock} will be carried out and
$\dd\big(w^{(l)},(\tilde{v}_0,\ldots,\tilde{v}_l)\big)\leq\lfloor (\sum_{i=0}^l d_i -1)/2\rfloor$.
In this case, Step~3 will return the correct data, that is, $w^{(l)}_0=v_0$.
\\
As a consequence, if Algorithm~\ref{A-DecBlock} has to turn to Case~a) or~b), it has detected an error in the sense that $\dd(\tilde{v},\cB)>\lfloor d/2 \rfloor$.
\end{theo}

\begin{proof}
Put $d^{(l)}:=\sum_{i=0}^l d_i-1$ for $l=0,\ldots,m$.
Then $d^{(m)}=d$ and by assumption $\dd(v,\tilde{v})\leq\lfloor d^{(m)}/2\rfloor$.
\\[.5ex]
Consider $l=m$.
There are two cases in which Algorithm~\ref{A-DecBlock} will have to proceed to $l-1$: either Step~2 is not executed at all or
the inequality in Step~3 is not satisfied.
In Part~1) and~2) below we will show that in both cases we obtain
\begin{equation}\label{e-nextstep}
  \dd\big((v_0,\ldots,v_{m-1}),(\tilde{v}_0,\ldots,\tilde{v}_{m-1})\big)\leq\lfloor d^{(m-1)}/2\rfloor,
\end{equation}
that is, the sent codeword and the received word satisfy the analogous error assumption on
the time interval $[0,m-1]$.
This will allow us to argue inductively.
In Part~3) we will show that there exists~$l$ for which the algorithm will return a result~$w^{(l)}_0$ and that
$w^{(l)}_0=v_0$.

1) Assume first that Step~2 has not been executed, hence $\dd(\tilde{v}_m,\,\cB_m)\geq\lfloor(d_m-1)/2\rfloor+1$.
Then we have in particular,
\begin{equation}\label{e-vmvmtilde}
    \dd(v_m,\tilde{v}_m)\geq\lfloor(d_m-1)/2\rfloor+1.
\end{equation}
Since $\dd(v,\tilde{v})\leq\lfloor d^{(m)}/2\rfloor$ this yields
\begin{equation}\label{e-vvtildem-1}
   \dd\big(v_0,\ldots,v_{m-1}),(\tilde{v}_0,\ldots,\tilde{v}_{m-1})\big)\leq\lfloor d^{(m)}/2\rfloor -\lfloor(d_m-1)/2\rfloor-1.
\end{equation}
Notice that $d^{(m)}=\sum_{i=0}^{m}d_i -1=\sum_{i=0}^{m-1}d_i+d_m-1=d^{(m-1)}+d_m$.
Going through all four cases of $d^{(m)}$ and $d_m$ being even or odd shows that
\[
   \lfloor d^{(m)}/2\rfloor -\lfloor(d_m-1)/2\rfloor-1\leq \lfloor d^{(m-1)}/2\rfloor.
\]
Hence\eqnref{e-vvtildem-1} implies\eqnref{e-nextstep}.
\\[.5ex]
2) Assume now that $\dd(\tilde{v}_m,\,\cB_m)\leq\lfloor(d_m-1)/2\rfloor$, hence Step~2 has been carried out,
and that
\begin{equation}\label{e-wmerror}
  \dd(w^{(m)},\tilde{v})>\lfloor d^{(m)}/2\rfloor,
\end{equation}
so that the assumption in Step~3 is not satisfied.
But then we may conclude\eqnref{e-vmvmtilde} again.
Indeed, if\eqnref{e-vmvmtilde} were not true, Reed-Solomon decoding of $\tilde{v}_m$ with respect to $\cB_m$ would
have returned the correct codeword~$v_m$ because $\lfloor(d_m-1)/2\rfloor$ is the error-correcting bound of the code~$\cB_m$.
In that case the associated message $\hat{u}\in\F^{(m+1)k}$ satisfying $\hat{u}G_{m,0}=v_m$ is unique
and would have resulted in $\hat{u}\hat{G}=v$.
Hence $w^{(m)}=v$, contradicting\eqnref{e-wmerror}.
Hence\eqnref{e-vmvmtilde} is true and as in~1), we arrive at\eqnref{e-nextstep}.

3) Suppose now that the algorithm proceeded to Step~2 for the value~$l$.
By the preceding discussion, see\eqnref{e-nextstep}, we then have
\begin{equation}\label{e-atl}
  \dd\big((v_0,\ldots,v_l),(\tilde{v}_0,\ldots,\tilde{v}_l)\big)\leq\lfloor d^{(l)}/2\rfloor.
\end{equation}
Assume furthermore that Step~2 has been executed and resulted in a word $w^{(l)}$ such that
$\dd\big(w^{(l)},(\tilde{v}_0,\ldots,\tilde{v}_l)\big)\leq\lfloor d^{(l)}/2\rfloor$.
Then\eqnref{e-atl} implies $\dd\big(w^{(l)},(v_0,\ldots,v_l)\big)\leq d^{(l)}$ and Remark~\ref{R-WeightAss} (applied to the code~$\cB$
in\eqnref{e-tildeB} with~$l$ instead of~$m$) shows that $w^{(l)}_0=v_0$ as desired.
\\
Finally, if, in the worst case, the algorithm has to proceed until $l=0$ we have, by\eqnref{e-atl},
\[
  \dd(v_0,\tilde{v}_0)\leq\lfloor d^{(0)}/2\rfloor=\lfloor (d_0-1)/2\rfloor.
\]
Since this is the error-correcting bound of the code $\cB_0$, decoding of~$\tilde{v}_0$ will take place and
will result in $w^{(0)}_0=v_0$.
This word will indeed be returned in Step~3 of the algorithm.
This completes the proof.
\end{proof}

Observe that Algorithm~\ref{A-DecBlock} might return a codeword $w^{(l)}_0$ in some cycle of Step~3) even if $\dd(\tilde{v},\cB)>\lfloor d/2\rfloor$.
The way the algorithm is formulated this potential decoding error will not be detected.
However, the overall Algorithm~\ref{A-Dec} might detect this situation by checking whether the received word~$\tilde{v}$ and the decoded word
$\hat{v}\in\cC$ satisfy
$\dd\big((\hat{v}_j,\ldots,\hat{v}_{j+m}),(\tilde{v}_j,\ldots,\tilde{v}_{j+m})\big)\leq\lfloor d/2\rfloor$ for all $j\in\N_0$, see\eqnref{e-error}.
One could, of course, extend Algorithm~\ref{A-DecBlock} by using the strategy of Case~a) in order to extend a partial codeword~$w^{(l)}$ to  full codewords
and checking whether one of those is within $\lfloor d/2\rfloor$ of~$\tilde{v}$.
But, as mentioned earlier, there is no guarantee that this strategy will find the closest codeword.

\section{Examples and Comparison to Other Decoding Algorithms}\label{S-ExaComp}
\setcounter{equation}{0}

We will first give some examples illustrating the algorithm.
Thereafter, we will compare the algorithm to other existing algorithms with respect to error-correcting capability and complexity.

Recall the data from Theorem~\ref{P-DCCC}.

\begin{exa}\label{E-F51}
Let $\F=\F_5$ and choose the primitive element~$\alpha=2$. Then $n=4$ and we choose $k=1$ and $m=2$.
Then $f=(x-1)(x-2)(x-4)=x^3+3x^2+4x+2$.
The $\F$-algebra automorphism~$\sigma$ is given by $\sigma(x)=2x$.
Thus, $\sigma(f)=3x^3+2x^2+3x+2$ and $\sigma^2(f)=4x^3+3x^2+x+2$ and we obtain
$G=G_0+G_1z+G_2z^2\in\F[z]^{1\times 4}$,  where
\[
  G_0=\begin{pmatrix}2&4&3&1\end{pmatrix},\
  G_1=\begin{pmatrix}2&3&2&3\end{pmatrix},\
  G_2=\begin{pmatrix}2&1&3&4\end{pmatrix}.
\]
We have to consider the block codes
\[
  \cB_0=\im G_0=\im\begin{pmatrix}2&4&3&1\end{pmatrix},\quad
  \cB_1=\im G_{1,0}=\im\begin{pmatrix}2&3&2&3\\2&4&3&1\end{pmatrix},\
\]
and
\[
  \cB_2=\im G_{2,0}=\im\begin{pmatrix}2&1&3&4\\2&3&2&3\\2&4&3&1\end{pmatrix}=\ker\begin{pmatrix}1\\1\\1\\1\end{pmatrix}.
\]
They have distance $d_0=4,\,d_1=3,$ and $d_2=2$, respectively. Hence $d=8$.
This is indeed the largest possible value for~$d$ satisfying\eqnref{e-d}; check, e.g., the codeword $(1,1,3)\hat{G}$, where~$\hat{G}$ is
as in\eqnref{e-BGhat}.
Thus, the algorithms~\ref{A-Dec}/\ref{A-DecBlock}
will reconstruct the sent codeword if no more than~$4$ errors have happened on any string of~$3$ consecutive
coefficients $(v_j,v_{j+1},v_{j+2})$.
Both the codes $\cB_0$ and $\cB_1$ can correct one error and $\cB_2$ cannot correct any errors.
The code~$\cB$ is given by
\begin{equation}\label{e-BGhat}
  \cB=\im \hat{G},\text{ where }
  \hat{G}=\left(\!\!\begin{array}{cccccccccccc}2&4&3&1&2&3&2&3&2&1&3&4\\0&0&0&0&2&4&3&1&2&3&2&3\\
           0&0&0&0&0&0&0&0&2&4&3&1\end{array}\!\!\right).
\end{equation}
Put
\[
  \tilde{G}=\begin{pmatrix}G_2&0&0\\G_1&G_2&0\end{pmatrix}=
  \left(\!\!\begin{array}{cccccccccccc}2&1&3&4&0&0&0&0&0&0&0&0\\2&3&2&3&2&1&3&4&0&0&0&0
            \end{array}\!\!\right)
\]
and for convenience write
\[
  \hat{G}_1=\begin{pmatrix}G_0&G_1\\0&G_0\end{pmatrix}=
  \left(\!\!\begin{array}{cccccccc}2&4&3&1&2&3&2&3\\0&0&0&0&2&4&3&1\end{array}\!\!\right).
\]
Let us consider the received word $\tilde{v}=(4031)+z(1130)+z^2(3210)+z^3(3213)+z^4(0100)$ and apply Algorithm~\ref{A-Dec} step by step.
\begin{Algolist}
\item $j=0$:
      Then $\tilde{V}=(\tilde{v}_0,\tilde{v}_1,\tilde{v}_2)=(403111303210)$ and $\hat{S}=0$.
      Hence $\tilde{w}=\tilde{V}$.
      Use of Algorithm~\ref{A-DecBlock}:
      \begin{algolist}
      \item $l=2$. We have $\tilde{w}_2\not\in\cB_2$. Since~$\cB_2$ cannot correct any errors, we go to
      \item $l=1$. Decoding $\tilde{w}_1=(1130)$ with respect to~$\cB_1$ yields $w^{(1)}_1=(12)G_{10}\in\cB_1$.
                   Hence no errors need to be corrected.
                   We compute $w^{(1)}=(12)\hat{G}_1=(24311130)$ and
                   $\dd\big(w^{(1)},(\tilde{w}_0,\tilde{w}_1)\big)=2\leq\lfloor (d_0+d_1-1)/2\rfloor$.
                   Thus Alg.~\ref{A-DecBlock} returns $w^{(1)}_0=(2431)$ and $\hat{x}_0=1$.
      \end{algolist}
      Alg.~\ref{A-Dec} returns $\hat{v}_0=(2431)$ and $\hat{u}_0=1$.
\item $j=1$:\
      $\tilde{V}=(\tilde{v}_1,\tilde{v}_2,\tilde{v}_3)=(113032103213),\ \hat{S}=(01)\tilde{G}=(232321340000),\
       \tilde{w}=(431311313213)$.
      Use of Algorithm~\ref{A-DecBlock}:
      \begin{algolist}
      \item $l=2$. $\tilde{w}_2\not\in\cB_2$.
      \item $l=1$. Decoding $\tilde{w}_1=(1131)$ w.r.t.~$\cB_1$ yields $w^{(1)}_1=(12)G_{10}=(1130)\in\cB_1$.
                   We compute $w^{(1)}=(12)\hat{G}_1=(24311130)$ and
                   $\dd\big(w^{(1)},(\tilde{w}_0,\tilde{w}_1)\big)=5\geq\lfloor (d_0+d_1-1)/2\rfloor$.
                   Hence we go to
      \item $l=0$. Decoding $\tilde{w}_0=(4313)$ w.r.t.~$\cB_0$ results in $w^{(0)}_0=2G_0=(4312)\in\cB_0$.
                   Since $\dd(w^{(0)}_0,\tilde{w}_0)=1\leq\lfloor(d_0-1)/1\rfloor$,
                   Alg.~\ref{A-DecBlock} returns $w^{(0)}_0=(4312)$ and $\hat{x}_0=2$.
      \end{algolist}
      Alg.~\ref{A-Dec} returns $\hat{v}_1=(4312)+(2323)=(1130)$ and $\hat{u}_1=2$.
\item $j=2$:\
      $\tilde{V}=(\tilde{v}_2,\tilde{v}_3,\tilde{v}_4)=(321032130100),\ \hat{S}=(12)\tilde{G}=(122042130000),\
      \tilde{w}=\tilde{V}-\hat{S}=(204040000100)$.
      Use of Algorithm~\ref{A-DecBlock}:
      \begin{algolist}
      \item $l=2$. $\tilde{w}_2\not\in\cB_2$.
      \item $l=1$. Decoding $\tilde{w}_1=(4000)$ w.r.t.~$\cB_1$ yields $w^{(1)}_1=(00)G_{10}=(0000)\in\cB_1$.
                   We compute $w^{(1)}=(00)\hat{G}_1=(00000000)$ and
                   $\dd\big(w^{(1)},(\tilde{w}_0,\tilde{w}_1)\big)=3\leq\lfloor (d_0+d_1-1)/2\rfloor$.
                   Alg.~\ref{A-DecBlock} returns $w^{(1)}_0=(0000)$ and $\hat{x}_0=0$.
      \end{algolist}
      Alg.~\ref{A-Dec} returns $\hat{v}_2=(0000)+(1220)=(1220)$ and $\hat{u}_2=0$.
\item $j=3$:
      $\tilde{V}=(\tilde{v}_3,\tilde{v}_4,\tilde{v}_5)=(321301000000),\ \hat{S}=(20)\tilde{G}=(421300000000),\
      \tilde{w}=\tilde{V}-\hat{S}=(400001000000)$.
      Use of Algorithm~\ref{A-DecBlock}:
      \begin{algolist}
      \item $l=2$. $\tilde{w}_2=(0000)=(000)G_{2,0}\in\cB_2$. Hence $w^{(2)}=0\in\F^{12}$  and
                   $\dd(w^{(2)},\tilde{w})=2\leq\lfloor (d_0+d_1+d_2-1)/2\rfloor$.
                    Alg.~\ref{A-DecBlock} returns $w^{(2)}_0=(0000)$ and $\hat{x}_0=0$.
      \end{algolist}
      Alg.~\ref{A-Dec} returns $\hat{v}_3=(0000)+(4213)=(4213)$ and $\hat{u}_3=0$.
\item $j=4$:
      $\hat{S}=(00)\tilde{G}=0$ and
      $\tilde{w}=\tilde{V}=(\tilde{v}_4,\tilde{v}_5,\tilde{v}_6)=(010000000000)$.
      Use of Algorithm~\ref{A-DecBlock}:
      \begin{algolist}
      \item $l=2$. $\tilde{w}_2=(0000)=(000)G_{2,0}\in\cB_2$. Hence $w^{(2)}=0\in\F^{12}$  and
                   $\dd(w^{(2)},\tilde{w})=1\leq\lfloor (d_0+d_1+d_2-1)/2\rfloor$.
                    Alg.~\ref{A-DecBlock} returns $w^{(2)}_0=(0000)$ and $\hat{x}_0=0$.
      \end{algolist}
      Alg.~\ref{A-Dec} returns $\hat{v}_4=(0000)$ and $\hat{u}_4=0$.
\item Thereafter, $\tilde{w}$ is always zero and the Algorithm returns only zeros.
\end{Algolist}
Thus, we found $\hat{u}=\hat{u}_0+\hat{u}_1z=1+2z$ and
$\hat{v}=\sum_{i=0}^3z^i\hat{v}_i=(2431)+z(1130)+z^2(1220)+z^3(4213)$.
As to be expected $\hat{v}=\hat{u}G$ is a codeword.
Moreover, $\dd(\tilde{v}_{[j,j+2]},\hat{v}_{[j,j+2]})\leq 4$ for all $j\geq0$.
Notice also that $\dd(\tilde{v},\hat{v})=6$ for the overall Hamming distance of the polynomial codewords.
Since, due to Theorem~\ref{P-DCCC}(2), $\dfree(\cC)=12$ this shows that $\hat{v}$ is a closest codeword
in~$\cC$ and by some straightforward considerations one can show that it is the unique closest codeword.
\end{exa}

The previous example resulted in a codeword~$v\in\cC$ that is no more than $\lfloor d/2\rfloor$ errors
apart from the received word~$\tilde{v}$ on any window of length~$N=m+1$, thus,~$v$ and~$\tilde{v}$
satisfy\eqnref{e-error}.
This is, of course, due to the fact that there does indeed exist such a codeword~$v$.
But even if no codeword satisfying\eqnref{e-error}
exists the algorithm might return a codeword without having to invoke Case~a) or~b) of Algorithm~\ref{A-DecBlock}.
This is illustrated in parts~(a) and~(b) of the following example.
Part~(c) shows that the algorithm does not necessarily return a codeword closest to~$\tilde{v}$ with respect to the
overall Hamming distance.

\begin{exa}\label{E-F51error}
Consider again the code from Example~\ref{E-F51}.
\begin{alphalist}
\item Let the received word be $\tilde{v}=(2000)+z(4004)+z^2(4000)+z^3(0431)$.
       Then the algorithms~\ref{A-Dec}/\ref{A-DecBlock} will return the codeword $\hat{v}=0$ because of the following:
       \begin{Algolist}
       \item The word $(\tilde{v}_0,\tilde{v}_1,\tilde{v}_2)$ has weight $4$ and thus $v=0\in\cB$ satisfies the error assumption in
             Theorem~\ref{T-DecBlock} and the algorithm returns $\hat{v}_0=0,\,\hat{u}_0=0$.
       \item The word $(\tilde{v}_1,\tilde{v}_2,\tilde{v}_3)$ is 4 errors apart from the codeword $(0,0,1)\hat{G}\in\cB$ and the algorithm returns
             $\hat{v}_1=0$ and $\hat{u}_1=0$.
       \item For $j\geq2$ the words $(\tilde{v}_j,\tilde{v}_{j+1},\tilde{v}_{j+2})$ have weight at most 4 and thus the algorithm returns
             zero.
       \end{Algolist}
       Hence the algorithm returns the codeword $\hat{v}=0$ and due to
       $\dd\big((\tilde{v}_1,\tilde{v}_2,\tilde{v}_3),(\hat{v}_1,\hat{v}_2,\hat{v}_3)\big)=6>\lfloor d/2\rfloor$
       one detects that the error assumption\eqnref{e-error} was not satisfied.
       Of course, this result implies that no codeword $v\in\cC$ satisfies\eqnref{e-error}.
       Again, by some straightforward computations one can show that~$\hat{v}=0$ is the closest codeword in~$\cC$ with respect to the
      overall Hamming distance.
\item Let the received word be $\tilde{v}=(2431)+z(1130)+z^2(0000)+z^3(0200)+z^4(4100)+z^5(0004)+z^6(0003)+z^7(0020)+z^8(0004)+z^9(3400)$.
      Then $\wt(\tilde{v}_j,\tilde{v}_{j+1},\tilde{v}_{j+2})\leq 4$ for all $j\geq1$ and $\wt(\tilde{v}_0,\tilde{v}_1,\tilde{v}_2)=7$.
      Thus, assuming the zero word has been sent the error assumption\eqnref{e-error} is satisfied for all~$j$ except for $j=0$.
      The algorithm will return the codeword
      $\hat{v}=(2431)+z(1130)+z^2(0032)+z^3(0230)+z^4(4100)+z^5(1004)+z^6(0023)+z^7(0320)+z^8(4024)+z^9(3421)
       =(1+2z+2z^2+z^3+4z^4+3z^5+3z^6+4z^7)G$,
      and may compute the values
      $\dd\big((\hat{v}_j,\hat{v}_{j+1},\hat{v}_{j+2}),(\tilde{v}_j,\tilde{v}_{j+1},\tilde{v}_{j+2})\big)=2,3,3,2,2,3,4,5,4,2$ for $j=0,\ldots,9$.
      By checking those distances, the algorithm detects that no codeword in~$\cC$ satisfies the error assumption\eqnref{e-error}.
      Notice also that $\dd(\tilde{v},0)=16$ while $\dd(\tilde{v},\hat{v})=10$.
      Again, by some lengthy, but straightforward computations one can show that~$\hat{v}$ is the unique closest codeword in~$\cC$.
\item Unfortunately, in general the error assumption\eqnref{e-error} does not imply that~$v\in\cC$ is a codeword closest to~$\tilde{v}$ with respect to the
      overall Hamming distance.
      For instance, for $\tilde{v}=(2400)+z(1100)+z^2(0000)+z^3(0230)+z^4(4100)+z^5(0000)+z^6(0023)+z^7(0320)+z^8(0000)+z^9(3400)$ we have
      $\wt(\tilde{v}_j,\tilde{v}_{j+1},\tilde{v}_{j+2})\leq 4$ for all $j\geq0$ and therefore the error assumption is satisfied for $v=0$ and
      the algorithm will return the zero word.
      However, in this case the codeword $\hat{v}$ from~(b) satisfies $\dd(\tilde{v},\hat{v})=12<\dd(0,\tilde{v})$.
      Again, one can show that~$\hat{v}$ is the unique closest codeword in~$\cC$.
\end{alphalist}
\end{exa}

We will close the paper with comparing the error-correcting capability and time complexity of our algorithm with existing algorithms
handling codes of comparable size.
In order to do so, let us first summarize the performance of our algorithm.
It is known~\cite[p.~247]{RoRu00} that Reed-Solomon decoding of an $[n,k]$ code has a time complexity of
$\cO\big(n(\log_2 n)^2\big)$, counting operations in the field $\F_q$, where $q\geq n+1$.
Using list decoding this complexity will grow by the factor~$l$, where~$l$ is the size of the list of
codewords produced by the algorithm~\cite[p.~255]{RoRu00}.
At each cycle of Step~1) Algorithm~\ref{A-Dec} essentially consists of invoking at most $m+1\leq n$ times Reed-Solomon
(or list) decoding of a Reed-Solomon code of length~$n$ (and dimension $jk,\;j=1,\ldots,m+1$) and thus has a
time complexity of $\cO\big((m+1)n(\log_2 n)^2\big)$ of operations in the field $\F_q=\F_{n+1}$.
It can correct up to $\lfloor d/2\rfloor$ errors occurring on any string $v_{j(m+1)},\ldots,v_{(j+1)(m+1)-1},\,j\in\N_0$,
of $m+1$ consecutive codeword blocks, and where~$d$ is as in\eqnref{e-ddfree}.

Let us now turn to the decoding algorithms mentioned in the introduction.

1) First of all, since Algorithm~\ref{A-Dec}/\ref{A-DecBlock} applies to a convolutional code $\cC\subseteq\F_q[z]^n$ of
degree $km$, see Theorem~\ref{P-DCCC}, and where the field size is $q=n+1$, Viterbi decoding for this particular convolutional code is,
in general, not feasible due to the high state space cardinality $(n+1)^{km}$.

2) Let us consider an $[n,k]$ Reed-Solomon block code, like $\cB_0$, for the encoding/decoding of the data stream
$v_0,v_1,v_2,\ldots$.
Hence, each block~$v_j$ is encoded/decoded independently.
This requires a time complexity of $\cO\big(n(\log_2 n)^2\big)$ for the decoding of each block and can correct up to
$\lfloor(n-k)/2\rfloor$ errors on any block, which means up to $(m+1)\lfloor(n-k)/2\rfloor$ errors on a string of $m+1$
consecutive blocks.
While this is, in general, larger than $\lfloor d/2\rfloor$, it only applies if no more than $\lfloor(n-k)/2\rfloor$
errors appear on a single block.
For instance, none of the coefficients of the received word~$\tilde{v}$ in Example~\ref{E-F51} could have been correctly decoded
because each $\tilde{v}_j$ is more than one error apart from the block code~$\cB_0$.

3) Suppose now that we use a (generalized) Reed-Solomon code of the size of $\cB$, that is, an $[(m+1)n,(m+1)k]$ RS code, in
order to encode/decode every string of $m+1$ consecutive blocks $v_{j(m+1)},\ldots,v_{(j+1)(m+1)-1},\,j\in\N_0$ independently.
This way we could correct up to $\lfloor(m+1)(n-k)/2\rfloor$ errors on any such string.
But this enhanced error-correcting capability comes with a significantly higher time complexity.
Indeed, the complexity goes up to $\cO\big((m+1)n,(\log_2(m+1)n)^2\big)$, and this is counting operations in a much
larger field with at least $(m+1)n$ elements.

4) In this part, we will compare our algorithm with a decoding algorithm designed for unit memory convolutional codes in
the thesis~\cite[Sec.~4.3]{Wei98}.
Consider a code with generator matrix $G_0+G_1z$, where $G_0,\,G_1\in\F^{k\times n}$.
Suppose $G_0,\,G_1$ generate block codes with distances~$\delta_0,\,\delta_1$, respectively.
Then the algorithm in \cite[Sec.~4.3]{Wei98} can correctly recover the sent codeword provided that a) no more than a total of
$t:=\lfloor (\delta_0+\delta_1-1)/2\rfloor$ errors occurred during the transmission, b) the degree of the sent codeword
(or an upper bound thereof) is known, and c) the block codes generated by $G_0,\,G_1$ can be decoded effectively.
Applying this to the code in Theorem~\ref{P-DCCC} with memory $m=1$, we obtain $\delta_0=\delta_1=n-k+1$, and therefore the
algorithm in~\cite{Wei98} can correct up to a total of $t=n-k$ errors occurring during the whole transmission.
It is based on decoding the Reed-Solomon codes generated by~$G_0,\,G_1$ and thus has a running time of
$\cO\big(n(\log_2 n)^2\big)$ for the decoding of each codeword block.
In contrast, Algorithm~\ref{A-Dec}/\ref{A-DecBlock} can correct up to $t'=\lfloor n-k-\frac{k-1}{2}\rfloor$ errors occurring on each
string of~$2$ consecutive codeword blocks, see\eqnref{e-ddfree}, and the running time is essentially the same.
Notice that for $k=1$ we have $t'=t$, making our algorithm significantly more suitable for this class of codes than the algorithm
proposed in~\cite{Wei98}.
As for general dimension~$k$, it is easy to see that $2t'\geq t$ (due to $k\leq n/2$, see Theorem~\ref{P-DCCC}) and
therefore our algorithm corrects, on each string of~$4$ consecutive blocks, at least as many errors as the total amount
corrected by~\cite{Wei98} -- as long as no more than~$t'$ errors happened on each half of that string.
We would also like to point out that the algorithm in~\cite{Wei98} needs the whole received word in order to perform
decoding, while our algorithm is iterative in the sense that it starts decoding as soon as the first~$2$ blocks have been received.
Of course, the algorithm in~\cite{Wei98} is applicable to any convolutional code as long as it has unit memory,
whereas our algorithm depends on the weight property described in\eqnref{e-d},\eqnref{e-d2} and is specifically designed
for the codes of Theorem~\ref{P-DCCC}, but requires a weaker assumption on the memory.

5) Finally, it remains to compare Algorithm~\ref{A-Dec}/\ref{A-DecBlock} with an algebraic decoding algorithm developed for
convolutional codes in~\cite{Ro99}.
That algorithm is based on an input/state/output description of the code in question, and its performance
may be summarized as follows (after adjusting to row vector notation):
Suppose the $k$-dimensional code $\cC=\im G\subseteq\F[z]^n$ of degree~$\delta$ has i/s/o description
\[
   x_{t+1}=x_tA+u_tB,\quad y_t=x_tC+u_tD,
\]
where $(u_t,y_t)\in\F^{k+(n-k)}$ is the sequence of codeword coefficients, $x_t\in\F^{\delta}$ is the state sequence,
and $(A,B,C,D)\in\F^{\delta\times\delta}\times\F^{k\times\delta}\times\F^{\delta\times(n-k)}\times\F^{k\times(n-k)}$.
Suppose $\Theta\in\N$ is such that the matrix $(C,\,AC,\,\ldots,\,A^{\Theta-1}C)$ has full row rank and that $T>\Theta$
and $d_1\in\N$ are such that the matrix
\begin{equation}\label{e-M}
   M:=\begin{pmatrix}B\\ BA\\ \vdots\\ BA^{T-1}\end{pmatrix}\in\F^{Tk\times\delta}
\end{equation}
has full column rank and $\ker M:=\{v\in\F^{Tk}\mid vM=0\}$ is a block code of distance at least~$d_1$.
Then the decoding algorithm in~\cite{Ro99} will return the sent codeword if at most
\begin{equation}\label{e-lambda}
   \lambda:=\min\big\{\lfloor (d_1-1)/2\rfloor,\, \lfloor T/(2\Theta)\rfloor \big\}
\end{equation}
errors occurred on any time window $[j,\, j+T-1],\,j\in\N_0$, that is, if any string of~$T$ consecutive codeword blocks
does not contain more than~$\lambda$ errors.

In the sequel we will show that Algorithm~\ref{A-Dec}/\ref{A-DecBlock} is better suited for decoding
our particular class of codes than the algorithm in~\cite{Ro99}.
Indeed, we will show that for our class of codes $\lambda\leq\lfloor m/2\rfloor$ and $T\leq m+1$.
Comparing this to\eqnref{e-ddfree} shows that Algorithm~\ref{A-Dec}/\ref{A-DecBlock} can correct significantly more
errors on intervals of length $m+1$.
Indeed, using that $m\leq n/k-1$ it is not hard to see that $\lfloor d/2\rfloor\geq m$, so that our algorithm can correct at
least twice as many errors.

Let us now turn to the details.
Recall the data given right before Theorem~\ref{P-DCCC} and fix the parameters and the encoder~$G$ as in that theorem.
Then the code $\cC=\im G$ has degree $\delta=mk$ and therefore,
in order for~$M$ in\eqnref{e-M} to have full column rank we need $T\geq m$ and for $\ker M$ to be a nontrivial code
we even need $T\geq m+1$.
Write $G=[Q,P]$, where $Q\in\F[z]^{k\times k}$.
We first note that the matrix~$Q$ is upper triangular.
Indeed, the polynomial $f\in\F[x]$ given in Theorem~\ref{P-DCCC} has degree $n-k$ and thus $\deg(x^lf)\leq n-1$ for
all $l=0,\ldots,k-1$.
As a consequence, we do not have to reduce modulo $x^n-1$ when computing in the quotient ring~$A$.
Since $x^l\mid (x^lf)$, we see that the first~$l$ entries of the vector $\fv(x^lf)\in\F^n$ are zero while the
$(l+1)$-st entry is nonzero (since~$f$ has nonzero constant term).
But then the same is true for $\fv\big(\sigma^j(x^lf)\big)$ because for any polynomial $g\in A$ we have
$\sigma(g)=g(\alpha^kx)$ and no reduction modulo $x^n-1$ is needed.
All this shows that the matrix~$Q$ is upper triangular and that the diagonal entries have degree~$m$.
With the aid of Theorem~\ref{P-DCCC}(1) this yields that $Q^{-1}P$ is a proper rational matrix.
Now we may use the controller canonical form of $Q^{-1}P$ in order to get an i/s/o representation of the code.
Using the method outlined in \cite[Sec.~6.4.1]{Kail80} (and transposing everything for row vector notation)
shows that the matrix $A\in\F^{mk\times mk}$ is upper block triangular with diagonal blocks of size $m\times m$ and
$B\in\F^{k\times mk}$ is upper block triangular with diagonal blocks of size $1\times m$.
Thus, collecting the last rows of each block in the matrix~$M$ in\eqnref{e-M} results in a submatrix
$M'\in\F^{T\times mk}$ in which the first $(k-1)m$ columns are zero and thus $\rank M'\leq m$.
As a consequence, since $T\geq m+1$ there exists a nonzero vector $v\in\F^{Tk}$ of weight at most~$m+1$ such that $vM=0$.
Thus, $d_1\leq\dist(\ker M)\leq m+1$ and the error correcting bound in\eqnref{e-lambda} satisfies
$\lambda\leq\lfloor m/2\rfloor$.
Using some more detailed considerations one can show that, likewise, every other input/output partition of the codewords
(that is, permuting the columns of~$G$ before splitting the matrix into $[Q,P]$)
along with an according i/s/o representation leads to the same error-correcting bound $\lambda\leq \lfloor m/2\rfloor$.
Summarizing, we may conclude that Algorithm~\ref{A-Dec}/\ref{A-DecBlock} is better suited for decoding
the class of codes defined in Theorem~\ref{P-DCCC} than the algorithm in~\cite{Ro99}.


\end{document}